\documentclass[copyright,creativecommons]{eptcs}
\providecommand{\event}{Gandalf 2011} 
\usepackage{breakurl} 
\usepackage{amsmath,amssymb}
\usepackage{relsize}
\usepackage{xspace}
\usepackage{macros,macrosadd}
\usepackage{graphicx}
\usepackage{gastex}
\usepackage{tikz}
\usetikzlibrary{arrows,automata,shapes.geometric,decorations.pathmorphing,backgrounds,positioning,fit,petri}
\usepackage{wrapfig}

\title{Opacity Issues in Games with Imperfect Information}
\author{Bastien Maubert
\institute{ENS Cachan - Bretagne\\
Bruz, 35170, France}
\email{bastien.maubert@irisa.fr}
\and
Sophie Pinchinat
\institute{S4, IRISA, Campus de Beaulieu\\
Rennes, 35042, France}
\email{sophie.pinchinat@irisa.fr}
\and
Laura Bozzelli
\institute{Technical University of Madrid (UPM)\\
28660 Boadilla del Monte, Madrid, Spain
}
\email{laura.bozzelli@fi.upm.es}
}

\def\titlerunning{Opacity Issues in Games with Imperfect Information}
\def\authorrunning{B. Maubert \& S. Pinchinat  \& L. Bozzelli}
\begin{document}
\maketitle

\newtheorem{theorem}{Theorem}
\newtheorem{lemma}[theorem]{Lemma}
\newtheorem{property}[theorem]{Property}
\newtheorem{corollary}[theorem]{Corollary}
\newtheorem{proposition}[theorem]{Proposition}
\newtheorem{definition}{Definition}
\newtheorem{remark}{Remark}
\newtheorem{example}{Example}
\newcommand{\theoname}[1]{{\bf (#1)}}
\newcommand{\theocite}[1]{{\bf \cite{#1}}}

\newenvironment{proof}{\trivlist\item[\hskip \labelsep {\bf Proof}\enskip]}%
{\unskip\nobreak\hskip 2em plus 1fil\nobreak%
\fbox{\rule{0ex}{1ex}\hspace{1ex}\rule{0ex}{1ex}}%
\parfillskip=0pt \endtrivlist}

\begin{abstract}
  We study in depth the class of games with opacity condition, which
  are two-player games with imperfect information in which one of the
  players only has imperfect information, and where the winning condition relies on
  the information he has along the play. Those games are relevant for
  security aspects of computing systems: a play is \emph{opaque}
  whenever the player who has imperfect information never ``knows''
  for sure that the current position is one of the distinguished
  ``secret'' positions. We study the problems of deciding the
  existence of a winning strategy for each player, and we call them
  the \emph{opacity-violate problem} and the \emph{opacity-guarantee
    problem}. Focusing on the
  player with perfect information is new in the field of games with imperfect-information because 
  when considering classical winning conditions it amounts to solving the underlying perfect-information game.
  We establish the $\EXPTIME$-completeness of both above-mentioned problems, showing that our winning condition
  brings a gap of complexity for the player with perfect information, and we exhibit the relevant
  \emph{opacity-verify problem}, which noticeably generalizes
  approaches considered in the literature for opacity analysis in
  discrete-event systems. In the case of blindfold games, this problem
  relates to the two initial ones, yielding the determinacy of
  blindfold games with opacity condition and the PSPACE-completeness of the three problems.
\end{abstract}

\section{Introduction}
\label{sec-introduction}
We described in \cite{maubert2009games} a class of two-player games
with imperfect information that we called \emph{games with opacity
  condition}. In these games, the players are \Intruder (for
``robber'') and \Defender (for ``guardian''). \Intruder has imperfect information as
opposed to \Defender who has perfect information. This asymmetric
setting is very relevant for the verification of open systems and all
the more for security aspects as it captures the intuitive picture of
an attacker having only a partial information against a system. 
The game model we consider relies on the classical imperfect-information arenas, as defined in {\it e.g.}
\cite{reif84,berwanger2008power}, but it is equipped with a subset of positions that denote confidential
information and that we call \emph{secrets}. 
We focus on the
opportunity for \Intruder to discover some secret, by introducing the
property of \emph{opacity}: a play is \emph{opaque} if, at each step
of the (infinite) play, the set of positions that are considered
possible by \Intruder does not consist of secrets only.
In games with opacity condition, the opacity property is the winning
condition for \Defender. Informally, \Intruder tries to force the game
to reach some point when he knows for sure that the current position
is a secret, whereas \Defender tries to keep \Intruder under
uncertainty.  Note that this winning condition can be seen as a
particular epistemic temporal logic
statement  \cite{halpern1989complexity} on an imperfect information
arena seen as an epistemic temporal model : this ETL formula is $G\neg
K_{\Intruder} secret$. However, to our knowledge the complexity of deciding the existence
of winning strategies for such winning conditions has never
been studied in depth. 

Our claim that games with opacity condition are natural and adequate
models for practical applications is all the more sustained by very
recent contributions of the literature
\cite{saboori08b,dubreil2008opacity}. These results mainly arise from the
analysis of discrete-event systems and their theory of control, and our games 
embed some problems studied in this domain, such as the verification of opacity. Our
abstract setting provided by the game-theoretical paradigm enables us to
focus on essential aspects of the topic, such as synthesizing
strategies, and to circumvent the complexity of the problems.

Not surprisingly,  games
with opacity condition are not determined  \cite{maubert2009games}. We therefore introduced two
dual problems: the \emph{opacity-violate problem} and the
\emph{opacity-guarantee problem}, that consist of deciding the
existence of a winning strategy, respectively for \Intruder and for
\Defender. The opacity-violate problem generalizes the strategy
problem in reachability games with imperfect information
\cite{reif84}, and so does the opacity-guarantee problem, but putting
the emphasize on the player who has perfect information and has the
complementary safety objective. The latter is, to our knowledge, never
been done, for the following reason. In two-player games with imperfect
information, when considering the existence of winning strategies for
a player, one can equivalently consider that the opponent has perfect
information (see \cite{reif84}).  Thus, when dealing with omega-regular winning
conditions in arenas where the imperfect information is asymmetric,
focusing on the player with perfect information
would amount to solve the underlying perfect-information game. Our
case is different : when considering \Defender's point of view, we
could indeed equivalently consider that \Intruder plays with perfect
information too, but we cannot give up the imperfect-information
setting because the definition of the winning condition itself relies
on \Intruder's information along the play. 

Additionally to the two aforementioned problems, we consider the
\emph{opacity-verify problem} as an intermediate problem: the question
here is to decide whether in a game with opacity condition, all
strategies of \Defender are winning. The choice of considering this
apparently weird problem is well motivated. Firstly, it is equivalent
both to the opacity-guarantee problem and to the complementary of the
opacity-violate problem for blindfold games; an immediate consequence
is the determinacy of blindfold games with opacity condition. And
secondly, it enables us to embed opacity issues in discrete-event
systems with a strong language-theoretic feature, addressed earlier in
the literature \cite{saboori08b,dubreil2008opacity}.

In this contribution, we consider the three problems of
opacity-violate, opacity-guarantee and opacity-verify, keeping in mind
that our main attention turns to the opacity-guarantee problem. It is
not hard to establish the $\EXPTIME$-completeness of the
opacity-violate problem, from a power-set construction inspired by
\cite{reif84} that amounts to solving a reachability
perfect-information game, and from the fact that it generalizes
imperfect-information games with reachability condition, known to be
$\EXPTIME$-complete \cite{reif84}. Regarding the opacity-guarantee
problem, we rely on an earlier power-set construction to reduce this
problem to a perfect-information game \cite{maubert2009games},
yielding $\EXPTIME$ membership. The EXPTIME-hardness result for
this problem, where the main player (\Defender) has perfect
information, was unknown until now and relies on a reduction from the
empty input string acceptance problem for linearly-bounded alternating
Turing machines. The key point is a pioneer encoding of configurations
by information sets. Concerning the opacity-verify problem, we prove
its $\PSPACE$-completeness, which for the lower bound relies on a
reduction similar to the one in \cite{de2006antichains}
from the universality problem for nondeterministic automata
\cite{hopcroft2006automata}. Interestingly, the opacity-verify problem
relates the two other problems for the particular case of
\emph{blindfold games}, in such a way that those games are determined.
We also show that the blindfold setting embraces the
language-theoretic approaches for opacity analysis in discrete-event
systems
\cite{saboori08b,dubreil2008opacity}.

The paper is organized as follows.  In Section~\ref{sec-games}, we
define games with opacity condition.  In
Section~\ref{sec-opviolguarpbs}, we present the opacity-guarantee
problem and the opacity-violate problem, and we establish their
$\EXPTIME$ complexity. We first recall the power-set constructions
from \cite{maubert2009games} yielding the upper bounds, then we show
the matching lower bounds. In Section~\ref{sec-blindfold}, we consider
the opacity-verify problem for blindfold games. In this setting, we
establish the determinacy and the $\PSPACE$ completeness of the three
opacity problems.
In Section~\ref{related}, we relate the opacity-verify problem to the
language opacity verification of
\cite{saboori08b,dubreil2008opacity}. In Section~\ref{sec-discussion},
we discuss complexity aspects of problems regarding \Defender's
winning strategies. We conclude in Section~\ref{sec-perspectives} by
giving some ideas on our current and future work.

\section{Games with opacity  condition}
\label{sec-games}

A \emph{game with opacity condition} over the alphabet $\Sigma$ and
the set of observations $\Gamma$ is an imperfect information game structure $A
=(V,\Delta,\mathrm{\obs},\mathrm{act},v_0,S)$ where
$V$ is a finite set of \emph{positions}, $\Delta:V \times \Sigma \to
2^V\backslash\emptyset$ is a \emph{transition function}, $\mathrm{\obs}: V \to \Gamma$ is
an \emph{observation function}, and $\mathrm{act}: \Gamma \rightarrow
2^\Sigma\backslash \emptyset$ assigns to each observation a non-empty
set of available actions, so that 
available actions are identical for observationally equivalent
positions. Finally, $v_0$ is the initial position, and the additional
element $S\subseteq V$ in the structure $A$ is a finite set of
\emph{secret positions}. 

In a game $A =(V,\Delta,\mathrm{\obs},\mathrm{act},v_0,S)$, the
players are \Defender and \Intruder. A play is an infinite sequence of
rounds, and in each round $i\geq 1$, \Intruder chooses an action $a_i
\in \mathrm{\act}(\mathrm{\obs}(v_{i-1}))$, \Defender chooses the new
position $v_i \in \Delta(v_{i-1},a_i)$, and \Intruder observes \obs$(v_i)$.
A \emph{play} in $A$ is an infinite sequence
$\rho=v_0a_1v_1\ldots \in v_0(\Sigma V)^{\omega}$ that results
from an interaction of \Intruder and \Defender in this game.

We now extend \obs{} to plays by letting
$\mathrm{\obs}(v_0a_1v_1a_2v_2\ldots):=v_0a_1\gamma_1a_2\gamma_2\ldots$
with $\gamma_i=\mathrm{\obs}(v_i)$ for each $i\geq 1$.  The imperfect
information setting leads \Intruder to partially observe a play $\rho$
as $\mathrm{\obs}(\rho)$. Note that since the initial position is a
part of the description of the arena, it is known by \Intruder.

For every natural number $k \in \mathbb{N}$ and play $\rho$, we denote by $\rho^k
\in v_0(\Sigma V)^k$ the $k$-th prefix of $\rho$, defined by
$\rho^k:= v_0 a_1  v_1  \ldots a_k  v_k$, with
the convention that $\rho^0 =  v_0$.  We denote by $\rho^+$ an
arbitrary prefix of $\rho$.

Since the information revealed to \Intruder is
based on observations, a strategy of \Intruder in $A$ is a mapping of
the form $\alpha: v_0(\Sigma\Gamma)^* \to \Sigma$ such that for any
play prefix $\rho^k$ ending in observation $\gamma$,
$\alpha(\mathrm{\obs}(\rho^k))\in\mathrm{\act}(\gamma)$.
On the contrary \Defender has perfect information on how the play progresses, so a
strategy of \Defender in $A$ is a mapping of the form $\beta:
v_0(\Sigma V)^*\Sigma \to V$ such that for any play prefix $\rho^k$ ending in
position $v$, for all $a$ in \act$(\mathrm{\obs}(v))$,
$\beta(\rho^ka)\in\Delta(v,a)$.

Given strategies $\alpha$ and $\beta$ of \Intruder and of
\Defender respectively, we say that a play $\rho = v_0 a_1 v_1 \ldots$ is 
 \emph{induced by }$\alpha$ if $\forall k\ge 1,\, a_k = \alpha(\mathrm{\obs}(\rho^{k-1}))$, and $\rho$ is 
 \emph{induced by }$\beta$ if $\forall k\ge 1,\, v_k =
 \beta(\rho^{k-1}a_k)$. We also note $\alpha \widehat{~} \beta$ the
 only play induced by $\alpha$ and by $\beta$.

In the following, an observation $\gamma$ might be interpreted as the
set of positions it denotes, namely $\mathrm{\obs}^{-1}(\gamma)$.

Let us fix a play $\rho=v_0 a_1v_1 a_2 v_2
\ldots$. Note that every $k$-th prefix of $\rho$ characterizes a unique
\emph{information set} $\infoset(\rho^k) \subseteq V$ consisting of
the set of positions that \Intruder considers possible in the game after $k$ rounds. 
Formally, information sets can be defined inductively as follows. 
\begin{definition}
  \label{defIS}
For every  play $\rho=v_0 a_1v_1 a_2 v_2\ldots$, we let 
$I(\rho^0):=\{v_0\}$ and $I(\rho^{k+1}):=\Delta(I(\rho^k),a_{k+1})\cap \mathrm{\obs}(v_{k+1})$, for $k \in \mathbb N$.
\end{definition}

We now define the opacity property:
\begin{definition}
  For a given set of secret positions $S\subseteq V$, a play $\rho$
  satisfies \emph{the opacity property for $S$}, or is
  $S$-\emph{opaque}, if: 
$$\forall k\in \mathbb N, I(\rho^k)\nsubseteq S$$
\end{definition}

Informally, the opacity condition means that \Intruder never knows
with certainty that the current position is a secret, because there is
always one of the positions he considers possible that is not a secret. 
In a \emph{game with
  opacity condition}, the opacity property is the winning condition
for \Defender, \emph{i.e} $S$-opaque plays are winning for \Defender,
and the other ones are winning for \Intruder.

\begin{remark}
  The definition of the arena and of the opacity condition are
  slightly different from the ones in \cite{maubert2009games} : originally
  \Intruder's aim was to reach a singleton information set. We
  introduce here the set of secret positions and define the winning
  condition accordingly because it makes these games closer to the
  intuition behind opacity. Anyway the results established in
  \cite{maubert2009games} still hold in this setting, and
  adapting the proofs is straightforward.
\end{remark}

\section{Opacity-violate and opacity-guarantee problems}
\label{sec-opviolguarpbs}

It is well known that perfect-information games are determined
\cite{martin75}, and that imperfect-information games are not
determined in general. We recall that a game is \emph{determined} if each position is winning
for one of the two players. 

We proved the following result in \cite{maubert2009games}:

\begin{theorem}
  \label{theo-non-determinacy}
  Games with opacity condition are not determined in general.
\end{theorem}

This result leads to introduce two dual problems. We remind that
$\alpha$ (resp.\ $\beta$) stands for a strategy of \Intruder (resp.\
\Defender).  We first consider \Intruder's point of view.

\begin{definition}
Given a game with opacity condition $A
  =(V,\Delta,\mathrm{\obs},\mathrm{\act},v_0,S)$, the \emph{opacity-violate problem} is to decide whether the following property holds:
  $$\exists \alpha, \forall \beta,\; \alpha\widehat{~}\beta \mbox{ is not }S\mbox{-opaque}$$
\end{definition}

We now consider \Defender's dual point of view.

\begin{definition}
\label{def-guarantee}
Given a game with opacity condition $A
  =(V,\Delta,\mathrm{\obs},\mathrm{\act},v_0,S)$, the \emph{opacity-guarantee problem} is to decide whether the following property holds:
  $$\exists \beta, \forall \alpha,\; \alpha\widehat{~}\beta \mbox{ is }S\mbox{-opaque}$$
\end{definition}

\begin{remark}
  It is important to comment on Definition~\ref{def-guarantee}  
  regarding the universal quantification over \Intruder's strategies. 
  As defined, this quantification ranges over observation based strategies
  only. The opacity-guarantee problem would however be equivalent if this quantification
  ranged over the wider set of perfect information strategies, as already 
  argumented by Reif in \cite{reif84} : along a play, \Intruder's possible 
  behaviors are not restricted by observation-based strategies.
\end{remark}

In the rest of this section we prove the following result:

\begin{theorem}
\label{complexity}
The opacity-violate and opacity-guarantee problems are EXPTIME-complete.
\end{theorem}

 In the following, we adopt the classic
convention that the size of a game is the size of its arena, {\it
  i.e.\ } the number of positions.

\subsection{Power-set constructions for upper bounds}
\label{sec-upbounds}

We recall the power-set constructions of \cite{maubert2009games} that lead to equivalently solve perfect information games.

We first address the opacity-violate problem. Since we consider the point of view of the player with
imperfect information, this problem is close to problems usually studied in games
with imperfect information. This is why we can easily rely on
previous work on the topic to study its complexity. We remind the construction from \cite{maubert2009games},
which is strongly inspired from the one described by Reif in \cite{reif84} :

Let $A =(V,\Delta,\mathrm{\obs},\mathrm{act},v_0,S)$ be a game with
opacity condition. We define a reachability perfect-information game
$\tildeA$, where the players are \tildeintruder and
\tildedefender\footnote{We use the superlative ``Super'' here because
  in general the winning strategies of \tildedefender do not reflect
  any winning strategy of \Defender in $A$. She has ``more
  power'' than \Defender.}.
A position of $\tildeA$ is either $I$ where $I$ is a reachable
information set in $A$ - it is a position of \tildeintruder{}
-, or $(I,a)$ where $I$ is a reachable information set in $A$ and $a \in
\mathrm{\act}(I)$ \footnote{$\mathrm{act}(I)$ makes sense because an information set
is always a subset of a single observation.} - it is a position of \tildedefender.

The game is played as follows. It starts in the initial position
$I_0:=\{v_0\}$ of \tildeintruder. In a position $I$,
\tildeintruder chooses $a \in \mathrm{act}(I)$ and moves to position
$(I,a)$. Next, let $O$ be the set of reachable observations from $I$
by $a$. \tildedefender chooses a next
 information set $\Delta(I,a)\inter\gamma$, where $\gamma$
ranges over $O$. In $\tildeA$, a play
$I_0(I_0,a_1)I_1(I_1,a_2)\ldots$ is winning for \tildeintruder if it
reaches a position of the form $I$ with $I\subseteq S$, otherwise it is
winning for \tildedefender.

\begin{theorem}
\label{theo-naive-algo-tilde}
\cite{maubert2009games} \Intruder has a winning strategy in $A$, if and only if, \tildeintruder has a winning strategy in the
perfect-information game $\tildeA$.
\end{theorem}

Due to nondeterminacy (Theorem~\ref{theo-non-determinacy}), the opacity-guarantee problem has to be studied on its own.
We remind the power-set
construction for the opacity-guarantee problem described in \cite{maubert2009games}, that leads to a 
safety perfect-information game
$\hatA$. In this game, unlike in $\tildeA$, we maintain an extra information on
how \Defender is playing in $A$. The players in $\hatA$
are \hatintruder\footnote{we use the superlative ``Super'' as,
  contrary to what \tildeintruder could do in the game
  $\tildeA$, \hatintruder can take advantage of the extra
  information.} and \hatdefender. A position in $\hatA$ is
either of the form $(I,v)$ where $I$ is a reachable information set in
$A$, and $v \in I$ - it is a position of \hatintruder{} -, or
of the form $(I,v,a)$ where $I$ is a reachable information set in
$A$, $v \in I$, and $a \in \mathrm{act}(I)$ - it is a position
of \hatdefender{}. The initial position is $(\{v_0\},v_0)$.
In position $(I,v)$, \hatintruder chooses $a \in \mathrm{\act}(I)$, and
moves to $(I,v,a)$. In position $(I,v,a)$, \hatdefender chooses $v'
\in \Delta(v,a)$ and moves to $(I',v')$ where $I'= \Delta(I,a)\inter
\mathrm{\obs}(v')$. In $\hatA$, a play
$(I_0,v_0)(I_0,v_0,a_1)(I_1,v_1)\ldots$ is winning for \hatintruder if
it reaches a position $(I,v)$ with $I\subseteq S$, otherwise it is
winning for \hatdefender.

\begin{theorem}
  \label{theo-naive-algo-hat}
  \cite{maubert2009games} \Defender has a winning strategy in $A$, if and only if, \hatdefender has a winning strategy in the
  perfect-information game $\hatA$.
\end{theorem}

It is well known that perfect-information reachability games and perfect-information safety games are solvable
in $\PTIME$. Since the constructions of $\tildeA$ and $\hatA$ involve a single exponential blow-up, it follows from 
Theorems~\ref{theo-naive-algo-tilde} and \ref{theo-naive-algo-hat} that the opacity-violate
and opacity-guarantee problems are in $\EXPTIME$.

\subsection{Matching lower bounds}
\label{sec-lowbounds}

We prove here that the opacity-violate and the opacity-guarantee problems are $\EXPTIME$-hard.

First, $\EXPTIME$-hardness of the opacity-violate problem is proved by a reduction from
 reachability imperfect-information
games of \cite{reif84}. Recall that a \emph{reachability imperfect-information game}
is a game of imperfect information $A =
(V,F,\Delta,\mathrm{\obs},\mathrm{act},v_0)$ over $\Sigma$ and $\Gamma$ with a distinguished set of
\emph{target observations} $F \subseteq \Gamma$ that \Intruder aims at reaching.

\begin{theorem}
\cite{reif84}
Solving reachability imperfect-information games is $\EXPTIME$-complete.
\end{theorem}

The reduction is straightforward. Let $A =
(V,F,\Delta,\mathrm{\obs},\mathrm{\act},v_0)$ be a reachability
imperfect-information game over $\Sigma$ and $\Gamma$. We define the
game with opacity condition $A' :=
(V,\Delta,\mathrm{\obs},\mathrm{\act},v_0,S)$ over $\Sigma$ and
$\Gamma$, where $S=\bigcup_{\gamma \in F}\gamma$.
It is easy to see that solving the reachability
imperfect-information game $A$ is equivalent to solving the
opacity-violate problem in the game $A'$ : a winning strategy for
\Intruder to reach $F$ in $A$ is also a winning strategy for \Intruder in
$A'$, and vice versa (remember that the information set is always a
subset of the current observation).

We now show that the opacity-guarantee problem is $\EXPTIME$-hard by a
polynomial-time reduction from the acceptance problem of the empty
input string for \emph{linearly-bounded alternating} Turing Machines
(TM) with a binary branching degree, which is $\EXPTIME$-complete
\cite{chandra1976alternation}. The key idea is to encode TM configurations by the
information sets.

In the rest of this section, we fix such a TM machine
$\mathcal{M}=(B,Q=Q_{\forall}\cup Q_{\exists}\cup
\{q_{acc},q_{rej}\},q_0,\delta)$, where $B$ is the input alphabet,
$Q_\exists$ (resp. $Q_{\forall}$) is the set of existential (resp.
universal) states, $q_0\in Q$ is the initial state, $q_{acc}\notin
Q_{\forall}\cup Q_{\exists}$ is the (terminal) accepting state,
$q_{rej}\notin Q_{\forall}\cup Q_{\exists}$ is the (terminal)
rejecting state, and $\delta: (Q_{\forall}\cup Q_{\exists})\times
B\rightarrow (Q \times B \times \{+1,-1\})\times (Q \times B \times
\{+1,-1\})$ is the transition function.  In each non-terminal step
(i.e., the current state is in $Q_{\forall}\cup Q_{\exists}$),
$\mathcal{M}$ overwrites the tape cell being scanned, and the tape
head moves one position to the left ($-1$) or right ($+1$).  Let $n$
be the size of $\mathcal{M}$ and $[n]=\{1,\ldots,n\}$. We assume that
$n>1$.

Since $\mathcal{M}$ is linearly bounded, we can assume that
$\mathcal{M}$ uses exactly $n$ tape cells when started on the
\emph{empty} input string $\varepsilon$. Hence, a  configuration (of
$\mathcal{M}$ over $\varepsilon$) is a word $C=w_1\, (q,b)\, w_2\in
B^*\cdot (Q\times B)\cdot B^*$ of length exactly $n$ denoting that the
tape content is $w_1\,b\,w_2$, the current state is $q$, and the tape
head is at position $|w_1|+1$. The initial configuration $C_{init}$ is
given by $(q_0,\textvisiblespace)\,\textvisiblespace^{n-1}$, where $\textvisiblespace$ is the blank
symbol. Moreover, without loss of generality, we assume that when started on $C_{init}$,
no matter what are the universal and existential choices,
$\mathcal{M}$ always \emph{halts} by reaching a terminal configuration
$C$, i.e. such that the associated state, written $q(C)$, is in
$\{q_{acc},q_{rej}\}$ (this assumption is standard, see
\cite{chandra1976alternation}). For a non-terminal configuration $C=w_1\, (q,b)\, w_2$
(i.e. such that $q\in Q_\exists\cup Q_\forall$), we denote by
$succ_L(C)$ (resp. $succ_R(C)$) the  successor of $C$ obtained by
choosing the left (resp. the right) triple in $\delta(q,b)$.  An
\emph{accepting computation tree} of $\mathcal{M}$ over $\varepsilon$
is a finite tree $T$ whose nodes are labeled by  configurations and
such that the root is labeled by $C_{init}$, the leaves are labeled by
accepting  configurations $C$, i.e.  $q(C)=q_{acc}$, each internal
node $x$ is labeled by a non-terminal  configuration $C$, and: (1)
if $C$ is existential (i.e., $q(C)\in Q_\exists$), then $x$ has
exactly one child whose label is one of the two successors of $C$, and
(2) if $C$ is universal (i.e., $q(C)\in Q_\forall$), then $x$ has
exactly two children corresponding to the two successors $succ_L(C)$
and $succ_R(C)$ of $C$. We construct a game with opacity condition
$A_{\mathcal{M}}$ such that \Defender has a
winning strategy in $A_{\mathcal{M}}$ if, and only if,  there is an accepting
computation tree of $\mathcal{M}$ over $\varepsilon$ (Theorem~\ref{correctness}). Hence,
$\EXPTIME$-hardness of the opacity-guarantee problem follows.

In the game $A_{\mathcal{M}}$, the tape content can be retrieved from
the current information set (of size $n$), and the remaining information
about the current configuration is available in each position of the
information set. A step of the machine is simulated by two rounds of
the game: in the first round, depending on whether the current state
is universal or existential, \Intruder simulates the universal choice
of the next configuration or \Defender simulates the existential
choice, and the second round simulates the updating of the
configuration of the machine.
 
\noindent Here, we describe the construction of the game
$A_{\mathcal{M}}=(V,\Delta,\mbox{\obs},\mbox{\act},v_0,S)$.
\begin{enumerate}
\item \noindent
  \text{\hspace{0cm}}$V=\{v_0,safe_L,safe_R,safe_{choice}\}\cup
  \bigl(([n]\times B)\times([n]\times Q\times
  B)\times\{L,R,choice\}\bigr)$.

  \item  $\mbox{\obs}: V \to \Gamma=\{\gamma_0,\gamma_{choice},\gamma_L,\gamma_R\}$ is defined by
    \begin{displaymath}
  \begin{array}{ll}
\text{\hspace{0.4cm}} \mbox{\obs}(v)  & = \left\{
  \begin{array}{ll}
    \gamma_0
    &   \,\,\,\,\,\textrm{ if }  v=v_0
    \\
    \gamma_L
    &   \,\,\,\,\,\textrm{ if } 
    v\in \{safe_L\}\cup \bigl(([n]\times  B)\times([n]\times Q\times B)\times\{L\}\bigr)
    \\
    \gamma_R
    &   \,\,\,\,\,\textrm{ if }  v\in \{safe_R\}\cup \bigl(([n]\times  B)\times([n]\times Q\times B)\times\{R\}\bigr)
    \\
    \gamma_{choice}
    &  \,\,\,\,\,\textrm{  }
    \text{otherwise.}
\end{array}
             \right.\\
\end{array}
\end{displaymath}
\item $\mbox{\act}: \Gamma \to\Sigma=\{\forall_L,\forall_R,\exists\}\cup
  B$ is defined by
    \begin{displaymath}
  \begin{array}{ll}
\text{\hspace{0.4cm}} \mbox{\act}(\gamma)  & = \left\{
  \begin{array}{ll}
 \Sigma
             &   \,\,\,\,\,\textrm{ if }  \gamma=\gamma_0
             \\
 \{\forall_L,\forall_R,\exists\}
             &   \,\,\,\,\,\textrm{ if }  \gamma=\gamma_{choice}
             \\
  B
       &  \,\,\,\,\,\textrm{  }
       \text{otherwise.}
\end{array}
             \right.\\
\end{array}
\end{displaymath}
 \item $S=([n]\times  B)\times([n]\times \{q_{rej}\}\times B)\times\{choice\}$.
\end{enumerate}
We delay the formal definition of $\Delta:V\times\Sigma\rightarrow
2^V\backslash\emptyset$ after informally describing the running of the
game.

A  configuration $C$ is encoded by an \emph{information set}
$I_{f}(C)$ of the form
$$\{((1,b_1),(i,q(C),b_i),f),\ldots,((n,b_n),(i,q(C),b_i),f)\}$$
\noindent where $f\in \{L,R,choice\}$, $i$ is the position of the tape
cell of $C$ being scanned, and for each $1\leq j\leq n$, $b_j$ is the
content of the $j$-th cell. For each $f\in \{L,R,choice\}$, $I_{f}(C)$
is called the \emph{$f$-code} of $C$, and during a play, the current
information set is of the form $I_{f}(C)$ for some reachable
configuration $C$ of the machine, unless \Intruder happened to have made
some \emph{deviating} move which does not simulate the dynamics of $\mathcal M$.
We capture this deviation by
making \Intruder lose: technically, the play enters one of the
\emph{safe} positions $safe_L,safe_R$, or $safe_{choice}$ that do not
belong to the set $S$ of secrets; then, once a safe position is reached,
only other safe positions can be reached, yielding \Defender to win,
whatever \Intruder does in the future. Note that for each $f\in
\{L,R\}$, $I_f(C)$ does not violate the opacity condition for
$S$, and $I_{choice}(C)$ violates the opacity condition for $S$
if, and only if,  $C$ is rejecting (i.e. $q(C)=q_{rej}$).  For all $q\in
Q_\exists\cup Q_\forall$ and $b\in B$, we denote by $\delta_{L}(q,b)$
(resp. $\delta_{R}(q,b)$) the left (resp. right) triple in
$\delta(q,b)$.  The behavior of $A_{\mathcal{M}}$ is as follows:

\begin{description}
\item{\it First round}: From the initial position $v_0$, whatever \Intruder
and \Defender choose,
the information set at the end of the first round is $I_{choice}(C_{init})$, the $choice$-code
of the initial  configuration. \vspace{0.1cm}

\item{\it The current information set is $I_{choice}(C)$ for some
    terminal configuration $C$}: If $C$ is rejecting, then
  $I_{choice}(C)\subseteq S$ and \Defender loses. Otherwise,
  $I_{choice}(C)\not\subseteq S$ and independently of the move of
  \Intruder, the play reaches a safe position $safe_{dir}$ for some
  $dir\in \{L,R\}$ and \Defender wins.
\end{description}
As we shall see, there remain only two cases, which in turn simulate a
complete step of $\mathcal{M}$.
\begin{description}
\item{\it The current information set is $I_{choice}(C)$ for some
    non-terminal configuration $C$}: \\Let
  $v=((k,b_k),(i,q(C),b_i),choice)$ be the current position
  (corresponding to some position in $I_{choice}(C)$). From $\mbox{\obs}(v)$,
  \Intruder can only choose actions in
  $\{\exists,\forall_L,\forall_R\}$. There are again two cases.
\begin{description}
\item{\it $C$ is existential (note that this information is contained
  in the position $v$)}. 
Moves $\forall_L$ and $\forall_R$ of \Intruder are deviating
and the play reaches one of the safe positions $safe_L$ or $safe_R$,
thus \Defender wins. If instead \Intruder's move is $\exists$, the
following move $dir\in \{L,R\}$ of \Defender aims at simulating the
existential choice of $\mathcal{M}$ in the configuration configuration
$C$. The reached position is then
$v'=((k,b_k),(i,q(C),b_i),dir)$. 

\item{\it $C$ is universal}. The move $\exists$ of \Intruder is
  deviating and the following move of \Defender can lead only to
  $safe_L$ or $safe_R$, which makes him win. Instead \Intruder's move
  $\forall_{dir} \in \{\forall_L,\forall_R\}$ simulates the universal
  choice of $\mathcal{M}$ in the configuration $C$. Next, \Defender's
  move is unique and leads to the position
  $v'=((k,b_k),(i,q(C),b_i),dir)$.
\end{description}
Whatever the type of the configuration $C$ was, by letting the
observation classes split positions with different values of
$dir$ (see the definition of $\mbox{\obs}$ above), the information set
after the move of \Defender becomes $I_{dir}(C)$, unless \Intruder's move was deviating.
\item{\it The current information set is $I_{dir}(C)$ with $dir\in
    \{L,R\}$, for some non-terminal configuration $C$}: \\Let
  the current position be $v=((k,b_k),(i,q(C),b_i),dir)\in I_{dir}(C)$, and let\\
  $\delta_{dir}(q(C),b_i)=(q_{dir},b_{dir},\theta_{dir})$. The value
  $j=i+\theta_{dir}$ represents the position of the cell being scanned
  in the next configuration $succ_{dir}(C)$; note that the value $j$
  is easily computable from the current position $v$. In order however
  to complete the step of the machine and to reach the information set
  $I_{choice}(succ_{dir}(C))$, the value of $b_j$ must be provided by
  the game. Therefore, we let $b_j$ be the only non-deviating move of
  \Intruder from position $v \in I_{dir}(C)$, among the possible moves
  in $B$.

  From position $v=((k,b_k),(i,q(C),b_i),dir)$, the above behavior is
  implemented as follows.  Let $b$ be the action chosen by \Intruder.
  If $k\notin\{i,j\}$, tape cell $k$ is unchanged by the step of the
  machine, hence the only possible move of \Defender leads to
  $((k,b_k),(j,q_{dir},b),choice)$. If $k=i$, tape cell $i$ is
  overwritten, hence the move of \Defender is unique and leads to
  $((i,b_{dir}),(j,q_{dir},b),choice)$.  Finally, if $k=j$, there are
  two cases. If $b=b_j$, then \Defender can only move to
  $((j,b_j),(j,q_{dir},b_j),choice)$ which updates the data for the
  next configuration $succ_{dir}(C)$, otherwise the move $b$ ($\neq
  b_j$) of \Intruder is deviating (and the play reaches a safe position).
\end{description}

We can now formally define the moves in $A_{\mathcal{M}}$, by letting $\Delta : V\times\Sigma\rightarrow
2^V\backslash\emptyset$ be:\\

 \noindent \textbf{Case $v=v_0$:}
   \begin{displaymath}
\Delta(v,a) =  \{((h,\textvisiblespace),(1,q_0,\textvisiblespace),choice)\mid h\in [n]\}
\end{displaymath}

\noindent \textbf{Case $v=safe_{choice}$:}
   \begin{displaymath}
  \Delta(v,a)  =
 \{safe_{dir}\mid dir\in \{L,R\}\}
\end{displaymath}

\noindent \textbf{Case $v=safe_{dir}$}, where $dir\in \{L,R\}$:
    \begin{displaymath}
   \Delta(v,a)   = \{safe_{choice}\}
\end{displaymath}

\noindent \textbf{Case $v=((h,b),(i,q,b'),choice)$:}
   \begin{displaymath}
  \begin{array}{ll}
\noindent \text{\hspace{0.0cm}} \Delta(v,a)  & = \left\{
  \begin{array}{ll}
 \{((h,b),(i,q,b'),dir)\mid dir\in \{L,R\}\}
             &   \,\,\,\,\,\textrm{ if }  a=\exists \text{ and }q\in Q_\exists
             \\
 \{((h,b),(i,q,b'),L)\}
             &   \,\,\,\,\,\textrm{ if }  a=\forall_L \text{ and }q\in Q_\forall
             \\
  \{((h,b),(i,q,b'),R)\}
             &   \,\,\,\,\,\textrm{ if }  a=\forall_R \text{ and }q\in Q_\forall
             \\
  \{safe_{dir}\mid dir\in \{L,R\}\}
       &  \,\,\,\,\,\textrm{  }
       \text{otherwise}
\end{array}
             \right.\\
\end{array}
\end{displaymath}

\noindent \textbf{Case $v=((h,b),(i,q,b'),dir)$}, where $dir\in
\{L,R\}$, $q\notin \{q_{rej},q_{acc}\}$, and
$\delta_{dir}(q,b')=(q_{dir},b_{dir},\theta_{dir})$:
   \begin{displaymath}
  \begin{array}{ll}
\text{\hspace{0.0cm}} \Delta(v,a)  & = \left\{
  \begin{array}{ll}
 \{((h,b),(i+\theta_{dir},q_{dir},a),choice)\}
             &   \,\,\,\,\,\textrm{ if }   a\in B \text{ and } h\notin \{i,i+\theta_{dir}\}
             \\
 \{((h,b_{dir}),(i+\theta_{dir},q_{dir},a),choice)\}
             &   \,\,\,\,\,\textrm{ if }   a\in B \text{ and }  h=i
             \\
  \{((h,b),(i+\theta_{dir},q_{dir},b),choice)\}
             &   \,\,\,\,\,\textrm{ if }   a=b\text{ and } h=i+\theta_{dir}
             \\
  \{safe_{choice}\}
       &  \,\,\,\,\,\textrm{  }
       \text{otherwise}
\end{array}
             \right.\\
\end{array}
\end{displaymath}
\noindent \textbf{Case $v=((h,b),(i,q,b'),dir)$}, where $dir\in \{L,R\}$ and $q\in \{q_{rej},q_{acc}\}$:
   \begin{displaymath}
   \Delta(v,a)   = \{((h,b),(i,q,b'),choice)\}
\end{displaymath}

This achieves the construction of the game $A_{\mathcal M}$ which
satisfies the following result:

\begin{theorem}\label{correctness}
  {\bf\cite{bozzelliMaubertPinchinat11a}} There is an accepting computation tree of $\mathcal{M}$ over
  $\varepsilon$ if, and only if, there is a winning strategy of \Defender
  in the game $A_{\mathcal{M}}$.
\end{theorem}

\section{Blindfold games with opacity condition}
\label{sec-blindfold}

We recall that a game with imperfect information is \emph{blindfold} if all positions have the same observation.

\begin{lemma}
  \label{blind}
  Let $A=(V,\Delta,\mathrm{\obs},\mathrm{\act},v_0)$ be a blindfold
  game with imperfect information over $\Sigma$ and
  $\Gamma=\{\gamma\}$. For every play prefix
  $\rho^n=v_0a_1v_1\ldots a_n v_n$,
  $I(\rho^n)=\Delta(\{v_0\},a_1\ldots a_n).$
\end{lemma}

The proof is trivial, by applying the definition
of the information set.

In blindfold games \Intruder cannot base the choice of his actions on anything
because he sees nothing of what \Defender does. So a strategy for 
\Intruder is just an infinite sequence of actions. More formally:

\begin{lemma}
\label{blind-obs}
  Let $A=(V,\Delta,\mathrm{\obs},\mathrm{\act},v_0)$ be a blindfold
  game with imperfect information over $\Sigma$ and
  $\Gamma=\{\gamma\}$, let $\alpha$ be a strategy for \Intruder, then
  there exists $a_1a_2a_3\ldots\in\Sigma^\omega$ such that for all
   strategies $\beta$ and $\beta'$ for \Defender,
  $\mathrm{\obs}(\alpha\widehat{~}\beta)=\mathrm{\obs}(\alpha\widehat{~}\beta')=v_0a_1\gamma a_2\gamma\ldots$
\end{lemma}

In the rest of this section we prove the following two theorems:

\begin{theorem}
\label{blind-determined}
Blindfold games with opacity condition are determined.
\end{theorem}

\begin{theorem}
\label{PSPACEc-OGOV}
For blindfold games with opacity condition, the opacity-guarantee problem and the
opacity-violate problem are $\PSPACE$-complete.
\end{theorem}

Both theorems are proved by considering a third problem: the
\emph{opacity-verify problem} which addresses the strong ability for
\Defender to win the game. We define this problem and establish its PSPACE-completeness in the
general setting of games with opacity condition and also in the
particular case of blindfold games (Theorem~\ref{complexity-verify}). We finally compare it
to the opacity-violate and opacity-guarantee problems for blindfold
games (Theorem~\ref{eq-g-v}).

\begin{definition}
Given a game with opacity condition $A
  =(V,\Delta,\mathrm{\obs},\mathrm{\act},v_0,S)$, the \emph{opacity-verify problem} is to decide whether the following property holds:
\begin{equation}
\label{eq-Ove}
\forall \beta, \forall \alpha,\;\alpha\widehat{~}\beta \mbox{ is }S\mbox{-opaque}
\end{equation}
\end{definition}
If Property~\eqref{eq-Ove} holds, any strategy $\beta$ of 
\Defender is a winning-strategy. Otherwise, there exists a play in the game that is not $S$-opaque.

\begin{theorem}
  \label{complexity-verify}
  The opacity-verify problem is $\PSPACE$-complete, even for blindfold games.
\end{theorem}

For the $\PSPACE$ membership, we design an algorithm that decides
whether there exists a losing play for \Defender, which is clearly
equivalent to deciding whether there exists a strategy of \Defender that is
not winning. 
The algorithm runs in $\NPSPACE$, hence in $\PSPACE$ \cite{Savitch70}, by nondeterministically choosing the
moves for \Intruder and \Defender, and by updating the current
information set of \Intruder at each round. Since information sets are
subsets of the set of positions, if there are $n$ positions, we need
$O(n)$ space to run this algorithm.
The $\PSPACE$-hardness of the opacity-verify problem results from a
reduction from the universality problem for a complete
nondeterministic finite automaton (NFA), known to be
$\PSPACE$-complete \cite{stockmeyer1973word}. 
This reduction was initially inspired by \cite{dubreilphd} but is in fact a variant of the one in \cite{de2006antichains}.

We recall that a NFA $\mathcal A=(Q,\Sigma,\Delta,Q_0,Q_f)$ is a
nondeterministic finite automaton with states $Q$, alphabet
$\Sigma$, transition relation $\Delta:Q\times\Sigma\rightarrow 2^Q$ and sets of (respectively)
initial and accepting states $Q_0$ and $Q_f$. A NFA $\mathcal A$ is complete if
for every state $q$ and letter $a$, $\Delta(q,a)\neq\emptyset$. The \emph{language}
$\mathcal L(\mathcal A)\subseteq \Sigma^*$ of $\mathcal A$ is the set of words $w\in \Sigma^*$ such that
$\Delta(Q_0,w)\cap Q_f\neq\emptyset$.
 The universality problem is
to decide whether $\mathcal A$ accepts all possible finite words, \emph{i.e} $\mathcal
L(\mathcal A)=\Sigma^*$.

Given a complete NFA $\mathcal A=(Q,\Sigma,\Delta,Q_0,Q_f)$, 
define the blindfold game with opacity condition $A_{\mathcal
  A}=(Q\cup\{q_0\},\Delta',\mathrm{\obs},\mathrm{\act},q_0,S)$ over $\Sigma$ and
$\Gamma=\{\gamma\}$, with $q_0\notin Q$, as follows:

$$\begin{array}{lcr}
S=Q\backslash (Q_f\cup\{q_0\})\hspace{0.5cm} & \mathrm{\act}(\gamma)=\Sigma \hspace{0.5cm} &
\forall q \in Q\cup\{q_0\}, \mathrm{\obs}(q)=\gamma 
\end{array}$$
$$ \forall a\in \Sigma,\Delta'(q,a) = 
\begin{cases}
  Q_0 & \mbox{if } q = q_0 \\
  \Delta(q,a) & \mbox{otherwise}
\end{cases}$$


\pagebreak

Since, firstly, $q_0$ is not reachable after the first move, secondly,
$\Delta'(q,a)=\Delta(q,a)$ for $q\neq q_0$ and finally,
$\Delta'(q_0,a)=Q_0$ for all $a$, we obtain from lemma \ref{blind} the
following corollary :

\begin{corollary}
  \label{blindlemma}
  For each play prefix in $A_{\mathcal A}$ of the form $\rho^n=q_0a_1\ldots a_nq_n$ ($n\geq 1$),
  $I(\rho^n)=\Delta(Q_0,a_2\ldots a_n)$.
\end{corollary}

One may note that the aim of the initial position $q_0$ is to
initialise \Intruder 's information set to $Q_0$ at the end of the
first round. 

\begin{proposition}
The NFA $\mathcal A$ is universal if, and only if, in  $A_{\mathcal A}$, every strategy of \Defender is winning.
\end{proposition}
\begin{proof}
We start with the right-left implication. Assume that every strategy is winning for \Defender. 
Take one strategy $\beta$, and take a word $w\in\Sigma^*$. Consider a
play $\rho$ in which \Intruder's first moves form the sequence of
actions $aw$, for some $a$ in $\Sigma$, and \Defender follows strategy
$\beta$. This is possible because the underlying automaton is
complete. Being $\rho$ induced by the winning strategy $\beta$, it is
$S$-opaque, hence in particular $I(\rho^{1+|w|})\nsubseteq S$.  By
Corollary~\ref{blindlemma} we obtain : $\Delta(Q_0,w)\nsubseteq S$,
which implies that there exists a position $q$ in $\Delta(Q_0,w)$
that is in $Q_f$, hence $\mathcal A$ accepts $w$. $\mathcal A$ is
universal.

For the other implication, suppose that $\mathcal A$ is universal. Let $\beta$ be
a strategy of \Defender, and let $\rho$ be a play induced by $\beta$. We prove that
$\rho$ is $S$-opaque. Let $n \in \mathbb N$.  If $n=0$,
$I(\rho^n)=\{q_0\}\nsubseteq S$. If $n>0$, there exists $w$ in
$\Sigma^*$ such that $I(\rho^n)=\Delta(Q_0,w)$
(Corollary~\ref{blindlemma}).  Since $\mathcal A$ is universal it accepts
$w$, hence $\Delta(Q_0,w)\cap Q_f\neq \emptyset$. So $I(\rho^n)\nsubseteq
S$, and this finishes the proof.
\end{proof}

\begin{theorem}
  \label{eq-g-v}
  In the setting of blindfold games with opacity condition, the opacity-verify
  problem, the opacity-guarantee problem and the complementary of the opacity-violate problem are equivalent.
\end{theorem}

\begin{proof}

Let $A=(V,\Delta,\mathrm{\obs},\mathrm{\act},v_0,S)$ be a
  blindfold game with opacity condition.
  It is clear that in general, $$\forall \beta,\forall
  \alpha,\;\alpha\widehat{~}\beta \mbox{ is }S\mbox{-opaque}\Rightarrow
  \exists \beta,\forall \alpha,\;\alpha\widehat{~}\beta \mbox{ is
  }S\mbox{-opaque}$$ We prove the converse in the case of blindfold
  games. Suppose that there exists a winning strategy $\beta$ for \Defender.
  We prove that any strategy $\beta'$ is also winning. 

  Let $\alpha$ be a strategy for \Intruder. Since $A$ is blindfold, by Lemma~\ref{blind-obs} we have that
  $\mathrm{\obs}(\alpha\widehat{~}\beta)=\mathrm{\obs}(\alpha\widehat{~}\beta')$, so 
  for every $n\in \mathbb{N}$, $I(\alpha\widehat{~}\beta'^{ \;n})=I(\alpha\widehat{~}\beta^n)\nsubseteq S$.

  So we have that the opacity-verify problem is equivalent to the
  opacity-guarantee problem in blindfold games. We now show that the
  opacity-verify problem is also equivalent to the
  complementary of the opacity-violate problem (decide whether
  $\forall \alpha, \exists \beta \mbox{ s.t. } \alpha\widehat{~}\beta \mbox{ is
  }S\mbox{-opaque}$ holds).

  Once again one implication is trivial : 
$$\forall \beta,\forall
  \alpha,\;\alpha\widehat{~}\beta \mbox{ is }S\mbox{-opaque}\Rightarrow
  \forall \alpha,\exists \beta,\;\alpha\widehat{~}\beta \mbox{ is
  }S\mbox{-opaque}$$

  Now the other way. Suppose that for any strategy $\alpha$ there is a
  strategy $\beta$ for \Defender such that $\alpha$ loses. Now take
  any couple of strategies $(\alpha,\beta')$.  We know that there
  exists a strategy $\beta$ such that $\alpha\widehat{~}\beta$ is
  $S$-opaque.  But we also know (Lemma~\ref{blind-obs}) that
  $\mathrm{\obs}(\alpha\widehat{~}\beta)=\mathrm{\obs}(\alpha\widehat{~}\beta')$
  because the game is blindfold, so once again for every $n\in
  \mathbb{N}$,
  $I(\alpha\widehat{~}\beta'^n)=I(\alpha\widehat{~}\beta^n)\nsubseteq
  S$.
\end{proof}

The idea behind this theorem is that in blindfold games with opacity
condition, the outcome of a play does not rely on \Defender's
behaviour but only on what \Intruder plays. Indeed, since he observes
nothing of what \Defender does, \Intruder's information set, and so
the winning condition, are only determined by the series of actions he
chooses. Thus, these games via a power-set construction can be seen as
(reachability) one-player games: each position is a reachable
information set $I$, at each step the unique player (\Intruder)
chooses an action $a\in act (I)$, where $I$ is the current position,
and moves to position $\Delta(I,a) $. Therefore, in blindfold games with opacity condition,
whether \Intruder has a winning strategy (\emph{i.e} a winning sequence
of actions), or \Defender wins whatever he does.

The determinacy of blindfold games with opacity condition
(Theorem~\ref{blind-determined}) is an immediate corollary of the
above Theorem~\ref{eq-g-v}. Also Theorem~\ref{PSPACEc-OGOV} results
from Theorems~\ref{eq-g-v} and \ref{complexity-verify}.

\section{Related work}
\label{related}

Opacity has mostly been studied in the framework of discrete-event
systems and their theory of control (\cite{saboori08b,dubreil2008opacity}). It
is both interesting and important to know to what extent the classical
problems in this field can be embedded into our games. We first
describe the discrete-event system setting, next we define the notion of opacity
in this framework. We finally propose a translation from the verification of opacity in this setting to 
the opacity-verify problem in games with
opacity condition.

First we recall that a \emph{a deterministic finite automaton (DFA)}
is a NFA $\mathcal A=(Q,\Sigma,\delta,q_0,Q_f)$ but with a unique
initial state $q_0$ and in which the transition relation $\delta : Q\times \Sigma \rightarrow 2^Q$ satisfies
$|\delta(q,a)|\leq 1$ for all states $q$ and input symbols $a$.

The problem of opacity is defined in \cite{dubreil2008opacity} with regards to
a LTS $G$ (labelled transition system, \emph{i.e} a DFA without
accepting states) and a confidential predicate $\phi$ over execution
traces of $G$, representable by a regular language $\mathcal
L_\phi\subseteq \Sigma^*$ where $\Sigma$ is the set of events of the
transition system. For convenience, we equivalently state it on a DFA
$\mathcal A_G^\phi$ representing the transition system together with
the secret predicate. The automaton $\mathcal A_G^\phi$ is simply the
synchronized product of $G$ with some complete DFA accepting $\mathcal
L_\phi$.  We denote by $\mathcal T(\mathcal A)\subseteq \Sigma^*$ the
set of execution traces of an automaton $\mathcal A$, and by $\mathcal
L(\mathcal A)$ the language accepted by $\mathcal A$, so we have that
$\mathcal T(\mathcal A_G^\phi)=\mathcal T(G)$ and $\mathcal L(\mathcal
A_G^\phi)=\mathcal T(G)\cap\mathcal L_\phi$. From now on, for a
DFA $\mathcal A$, a state $q$ and $w\in \mathcal T(\mathcal A)$,
$\delta(q,w)$ shall denote the only state it contains.

We consider a subset of events $\Sigma_a\subseteq\Sigma$ which denotes
the observation capabilities of a potential attacker of the system,
and we let $P_{\Sigma_a}$ be the \emph{projection} function from $\Sigma^*$ to
$\Sigma_a^*$. Two words $w$ and $w'$ are \emph{observationally
  equivalent} if $P_{\Sigma_a}(w)=P_{\Sigma_a}(w')$.  We denote by
$[w]_a=P_{\Sigma_a}^{-1}(P_{\Sigma_a}(w))$ the set of words in
$\Sigma^*$ that are observationally equivalent to the word $w$ with
regard to $\Sigma_a$.

\begin{definition}
  $\mathcal L_\phi$ is \emph{opaque} w.r.t. $\mathcal T(G)$ and $\Sigma_a$ if
  $$\forall w \in \mathcal T(G), [w]_a\cap \mathcal T(G) \nsubseteq \mathcal L_\phi$$
\end{definition}

This means that $\mathcal L_\phi$ is opaque w.r.t. $\mathcal T(G)$ and $\Sigma_a$ if, and
only if, whenever an execution trace of $G$ verifies the confidential
predicate $\phi$ there exists another possible execution trace
observationally equivalent that does not verify $\phi$.

We take an instance of the opacity verification problem, $\mathcal A_G^\phi=(Q,\Sigma,\delta,q_0^G,Q_f)$, and
we describe the construction of the game with opacity condition $A_G^\phi$ such that the following holds.

\begin{theorem}
\label{reduc-verify}
Verifying that $\mathcal L_\phi$ is opaque w.r.t $\mathcal T(G)$ and
$\Sigma_a$ is equivalent to deciding the opacity-verify problem in
$A_G^\phi$.
\end{theorem}

The construction starts from $\mathcal A_G^\phi$ where transitions labelled by events in
$\Sigma\backslash\Sigma_a$ are turned into $\epsilon$-transitions. Then we
remove those $\epsilon$-transitions as described in \cite{hopcroft2006automata} by taking the
$\epsilon$-closure of the transition function, and we obtain the $\epsilon$-free
nondeterministic finite automaton 
$\mathcal A^{\epsilon}=(Q,\Sigma_a,\Delta^\epsilon,Q_0^\epsilon,Q_f)$.

In this automaton, transitions are all labelled by observable events.
One should think of the nondeterminism in this automaton as the uncertainty  the attacker 
has concerning the behaviour of the system. More precisely, she does not know when an observable event 
is triggered whether the system takes ``invisible'' transitions or not, may it be
before, after, or both before and after the observable one.

We need the following lemma, which is a mere consequence of the construction :

\begin{lemma}
\label{lemma2}
  $$\forall w\in\Sigma_a^*, \Delta^\epsilon(Q_0^\epsilon,w)=\{\delta(q_0^G,w')\;|\; w'\in [w]_a\cap\mathcal T(G)\}$$
\end{lemma}

We can now define the game $A_G^\phi=(V,\Delta,\mathrm{\obs},\mathrm{\act},v_0,S)$ over $\Sigma'=\{\surd\}$ 
and $\Gamma=\{\gamma_x\;|\;x\in\Sigma_a\}\cup\{\gamma_\epsilon\}$:
\begin{itemize}
  \item $V=Q\times\Sigma_a\cup Q_0^\epsilon\times\{\epsilon\}\cup\{v_{init}\}$.
  \item $\Delta(v,\surd)=\begin{cases}\{(q',y)\;|\;y\in\Sigma_a, q'\in
      \Delta^\epsilon(q,y)\} & \mbox{if }v=(q,x)\\
      \{(q,\epsilon)\;|\;q\in Q_0^{\epsilon}\} & \mbox{if }v=v_{init}\end{cases}$
  \item $\forall (q,x)\in V,\;\mathrm{\obs}((q,x))=\gamma_x$, and $\mathrm{\obs}(v_{init})=\gamma_\epsilon$
  \item $\forall v\in V,\;\mathrm{\act}(v)=\{\surd\}$
  \item $S=\{(q_f,x)\;|\;q_f\in Q_f, x\in\Sigma_a\cup\{\epsilon\}\}$\hspace{1cm}and\hspace{1cm}$v_0=v_{init}$
\end{itemize}

\begin{remark}
  Without loss of generality we can assume that in every state $q$ of $\mathcal A^\epsilon$
there exists an event $y$ in $\Sigma_a$ such that $\Delta^\epsilon(q,y)$ is not empty. So
in every position $(q,x)$ in $V$, $\Delta((q,x),\surd)$ is not empty, and the game can always continue.
\end{remark}

In this game, \Intruder is passive. He only observes \Defender, who
simulates the system $G$. If the game is in position $(q,x)$, it
represents that we are in state $q$ in the system $G$, and that the
last visible event was $x$ (if $x=\epsilon$, no observable event
happened yet). \Intruder observes $\gamma_x$, \emph{i.e} the only information
he gains during a play is the sequence of visible events.
When \Defender plays, he chooses a visible event $y$ and a state
reachable from $q$ through $y$ in $\mathcal A^\epsilon$, which can be
seen as choosing as many invisible transitions in $G$ as he wishes,
plus one visible amongst them, $y$. We shall note $\alpha_\surd$ the
only possible strategy for \Intruder, which is to always play $\surd$.

$v_{init}$ is the initial position, that can never be reached after
the first move. It is used to initialize \Intruder's information set
to $Q_0^\epsilon\times\{\epsilon\}$ (these are the only reachable
positions from $v_{init}$, and they have the same observation,
$\gamma_\epsilon$). This represents the set of states in $G$ that are
reachable before any observable transition is taken.

We start the proof of Theorem~\ref{reduc-verify} by establishing this central lemma.

\begin{lemma}
  \label{lemma1}
  Let $\rho^{n+1}=v_{init}\surd (q_0,\epsilon)\surd (q_1,x_1)\ldots \surd (q_n,x_n)$ be a prefix of a play, with $n\geq 0$.
  Then $\{q\;|\;(q,x_n)\in I(\rho^{n+1})\}=\Delta^\epsilon(Q_0^\epsilon,x_1\ldots x_n)$ and for all $(q,x)$ in 
  $I(\rho^{n+1})$, $x=x_n$.
\end{lemma}

\begin{proof}

The latter fact is obvious, from the definition of observations. Considering the former fact, we prove it by induction on $n$. 

\begin{description}
\item[$n=1$ :]
  $I(\rho^1)=\Delta(\{v_{init}\},\surd)\cap\gamma_\epsilon=\{(q_0,\epsilon)\;|\;q_0\in
  Q_0^\epsilon\}$, so
 $\{q\;|\;(q,\epsilon)\in I(\rho^1)\}=Q_0^\epsilon=\Delta^\epsilon(Q_0^\epsilon, \epsilon)$

\item[$n+1$ :]~\\ $ \begin{array}{lclr}
    \{q\;|\;(q,x_{n+1})\in I(\rho^{n+2})\} & = & \{q\;|\;(q,x_{n+1})\in \Delta(I(\rho^{n+1}),\surd)\cap\mathrm{\obs}((q_{n+1},x_{n+1}))\}\\
    & = & \{q\;|\;(q,x_{n+1})\in \Delta(I(\rho^{n+1}),\surd)\}\\
    & = & \{q\;|\;\exists (q',x_n)\in I(\rho^{n+1}),q\in \Delta^\epsilon(q',x_{n+1})\}\\
    & = & \{q\;|\;\exists q'\in \Delta^\epsilon(Q_0^\epsilon,x_1\dots x_n), q\in \Delta^\epsilon(q',x_{n+1})\}\\
    & = & \Delta^\epsilon(Q_0^\epsilon,x_1\ldots x_{n+1})
\end{array}$
\end{description}
\end{proof}
 
We move on to the proof of Theorem~\ref{reduc-verify}. 
Suppose that every strategy $\beta$ is winning for \Defender. We prove that $\mathcal L_\phi$ is
opaque w.r.t $\mathcal T(G)$ and $\Sigma_a$. Take a word $w$ in $\mathcal T(G)$. There exists a prefix of a play
$\rho^{n+1}=v_{init}\surd (q_0,\epsilon)\surd (q_1,x_1)\ldots \surd (q_n,x_n)$ such that $x_1\ldots x_n=P_{\Sigma_a}(w)$.
So there exists a strategy $\beta$ such that $\alpha_\surd\widehat{~}\beta^{n+1}=\rho^{n+1}$. With lemma~\ref{lemma1} and
\ref{lemma2} we have that $\{q\;|\;(q,x_n)\in I(\rho^{n+1})\}=\{\delta(q_0^G,w)\;|\;w\in [x_1\ldots x_n]_a\cap\mathcal
T(G)\}$. Since $\beta$ is winning, $\{q\;|\;(q,x_n)\in I(\rho^{n+1})\}\nsubseteq Q_f$, so 
there exists $w'$ in  $[x_1\ldots x_n]_a\cap\mathcal T(G) = [w]_a\cap\mathcal T(G)$ such that
$\delta(q_0^G,w')\notin Q_f$. This implies that $[w]_a\cap\mathcal T(G)\nsubseteq \mathcal L_\phi$. 

Now suppose that
$\mathcal L_\phi$ is opaque w.r.t $\mathcal T(G)$ and take $\beta$ a
strategy for \Defender in $A_G^\phi$, we prove that $\beta$ is
winning. Let $\rho_\beta=\alpha_\surd\widehat{~}\beta$ be the only possible
play induced by $\beta$. Take a prefix
$\rho_\beta^{n+1}=v_{init}\surd (q_0,\epsilon)\surd (q_1,x_1)\ldots \surd (q_n,x_n)$ of
this play.  By Lemma~\ref{lemma1} and \ref{lemma2} again,
$\{q\;|\;(q,x_n)\in I(\rho_\beta^{n+1})\}=\{\delta(q_0^G,w)\;|\;w\in [x_1\ldots x_n]_a\cap\mathcal
T(G)\}$. Since an information set is never empty, there exists $w$ in
$[x_1\ldots x_n]_a\cap\mathcal T(G)$, and because $\mathcal L_\phi$ is
opaque w.r.t $\mathcal T(G)$, $[x_1\ldots x_n]_a\cap\mathcal
T(G)\nsubseteq \mathcal L_\phi$. So there exists $w'$ in $[x_1\ldots x_n]_a\cap\mathcal T(G)$
such that $\delta(q_0^G,w')=q\notin Q_f$, hence $(q,x_n)\notin S$ and $I(\rho_\beta^n)\nsubseteq S$.
$\beta$ is winning.

\section{Discussion on complexity}
\label{sec-discussion}

Solving safety games with perfect-information is in PTIME, and solving parity games with
perfect information is known to be in $NP\cap\mbox{co-}NP$
\cite{jurdzinski1998deciding}. However we have seen that deciding
whether \Defender, who has perfect-information, has a winning strategy
in a game with opacity condition, is EXPTIME-complete, even if we let \Intruder play
with perfect-information (in the sense that his strategies are based on actual prefixes of plays instead of
their observation). So the gap between deciding the existence of a winning strategy
for a player in perfect-information games and for \Defender in a game
with opacity condition does not come from the fact that \Intruder has
imperfect information, but rather from the nature of the winning
condition itself, which is based on the notion of information set, and forces \Defender to keep track 
of what \Intruder's information set along the game is.

Similarly, verifying that a finite-state strategy is winning in a safety
perfect-information game can be done in $\PTIME$, whereas we have shown in \cite{bozzelliMaubertPinchinat11a} 
that in games with opacity condition, deciding whether a finite-state (and even memoryless) strategy of
\Defender is winning is PSPACE-complete in the size of the arena and
the memory of the strategy (we
define in a classic way the size of the memory of a strategy as the number of states of
an I/O automaton realizing the strategy \cite{dziembowski1997much}). The idea is that one has to check
that the strategy is winning not in all positions, but in all
information sets.  Concerning the size of the memory needed for
\Defender's strategies, we know that an exponential memory is
sufficient because if there is a winning strategy there is a
memoryless one in the powerset construction. The lower bound for the
needed memory is still an open problem.

\section{Conclusion and perspectives}
\label{sec-perspectives}

Following \cite{maubert2009games}, we have extended the study of
games with opacity condition. The opacity condition is an atypical
winning condition in imperfect information arenas aiming at capturing
security aspects of computer systems. Since games with opacity
condition are not determined in general, two dual problems need being
considered: the opacity-violate problem and the opacity-guarantee
problem, focusing on the player who has imperfect information and on the
player who has perfect information respectively. 
The latter problem is usually equivalent to solving the underlying
perfect information game, which explains why it has never been considered;
but the fact that our winning condition is based on information sets
makes the problem relevant.  For both problems, simple power-set
constructions apply to convert such games into perfect information
ones, that can be solved in polynomial time, hence their upper bound
is $\EXPTIME$. On the contrary, the matching $\EXPTIME$ lower bound
for the opacity-guarantee problem, where the main player has perfect
information, was unknown until now and relies on an elegant reduction
from the empty input string acceptance problem for linearly-bounded
alternating Turing machines. The key point is to encode 
configurations by information sets. The reduction and its
correctness proof are very technical, but we could provide an
intuitive informal description.

Finally, we focused on the particular case of blindfold games which
offers specific results such as determinacy
(Theorem~\ref{blind-determined}) and $\PSPACE$-complete complexities
(Theorem~\ref{PSPACEc-OGOV}). The main tool to obtain these results is
the opacity-verify problem which addresses the question whether any
strategy of \Defender is winning. The fact that blindfold games with
opacity condition can be seen as one-player games makes this problem
relevant and explains why it is equivalent to the opacity-guarantee
problem and to the complement of the opacity-violate problem in the blindfold setting, as we established. 
We also proved that
it is $\PSPACE$-complete, by providing a $\PSPACE$ algorithm and a
reduction from the nondeterministic finite automata universality
problem. The opacity-verify problem is all the more interesting to
consider that it naturally demonstrates how the paradigm of opacity
condition embraces opacity issues investigated in the recent
literature of Control Theory \cite{saboori08b,dubreil2008opacity}.

Games with opacity condition open a novel field in the theoretical
aspects of games with imperfect information by putting the emphasis on
the player who has perfect information. From this point of view, plethora
of questions need being addressed, among which their connection with
language-theoretic issues (the synchronizing/directing word problem
\cite{Cerny64,pin,CernyPirickaRosenauerova64}, controller synthesis to enforce the opacity of a
language \cite{dubreil2008opacity}), their logical
foundations, and their algorithmic aspects. \\

\section*{Acknowledgements} We are very grateful to the reviewers for relevant comments and suggestions that significantly helped 
in highlighting the conceptual content of the paper.


\end{document}